\documentclass{jfm}
\usepackage{graphicx}
\usepackage{newtxtext}
\usepackage{newtxmath}
\usepackage{natbib}
\usepackage{booktabs}
\usepackage[hidelinks]{hyperref}
\graphicspath{{figs/}}
\newcommand{\Ret}{\ensuremath{Re_\tau}}
\newcommand{\dV}{\ensuremath{\Delta V}}
\newcommand{\uref}{\ensuremath{u_{\mathrm{ref}}}}
\newcommand{\nut}{\ensuremath{\nu_t}}
\newcommand{\nutan}{\ensuremath{\nu_{\mathrm{tan}}}}
\newcommand{\nueff}{\nutan}
\newcommand{\epsU}{\ensuremath{\epsilon_U}}
\newcommand{\sym}{\operatorname{sym}}
\newcommand{\D}{\ensuremath{\mathcal{D}_1}}
\newcommand{\DD}{\ensuremath{\mathcal{D}_2}}
\newcommand{\amp}{\ensuremath{\mathcal{A}}}

\newtheorem{theorem}{Theorem}

\title{Well-posedness of neural turbulence closures and tangent dissipation}

\author{Zhen Zhang\aff{1}
 \and George Em Karniadakis\aff{1}\corresp{\email{george\_karniadakis@brown.edu}}}

\affiliation{\aff{1}Division of Applied Mathematics, Brown University, Providence, RI 02912, USA}

\shorttitle{Well-posedness of neural turbulence closures}
\shortauthor{Z. Zhang and G. E. Karniadakis}

\newcommand{\mainref}[1]{\ref{#1}}
\makeatletter
\def\@maketitle#1{%
  \par\vspace*{8pt}
  {\raggedright\Large\bfseries\@title\par}
  \vspace{12pt}
  {\raggedright\large\@author\par}
  \vspace{6pt}
  {\raggedright\small\@affiliation\par}
  \vspace{16pt}
}
\makeatother
\hypersetup{pdftitle={Well-posedness of neural turbulence closures and tangent dissipation},
 pdfauthor={Zhen Zhang and George Em Karniadakis}}

\begin{document}
\maketitle
\pagestyle{myheadings}
\thispagestyle{myheadings}
\markboth{Z. Zhang and G. E. Karniadakis}{Well-posedness of neural turbulence closures}

\begin{abstract}
A neural turbulence closure defines a new boundary-value problem,
$R(U)=N(U)+F(U)=0$, with a coupled Jacobian $J(U)=N'(U)+F'(U)$, where $N$ is the
original mean-flow operator and $F$ the learned closure. We establish two consequences
of global tangent dissipation. For a monotone original operator, a positive uniform
margin supplied by the original operator and closure together guarantees existence,
uniqueness and a global inverse-sensitivity bound relating a posteriori solution error
to the a priori residual.
For a general original operator, a dissipative closure cannot worsen tangent
dissipation, but this alone does not guarantee uniqueness. Tangent dissipation depends
on both diffusion and reaction. We study two complementary ways to promote it:
(1) an exact-integral construction enforcing non-negative tangent diffusion while leaving
reaction unconstrained, and (2) a penalty on tangent-reaction violations at sampled states.
Tangent diffusion enters the Jacobian, and non-negative secant eddy viscosity alone does
not control its coercivity. We conduct tests with channel flow at $\Ret=180$--$5200$, which provides a strongly
monotone baseline. Both constrained closures reach accurate solutions in all 50
training-seed/Reynolds-number cases. At $\Ret=1000$, we conduct tests with 10,000 starts for one fixed
network per closure and we find one root for each constrained closure and multiple roots for
the other closures. Although this does not prove uniqueness, it provides strong
empirical evidence for uniqueness of the tested constrained closures.
At $\Ret=5200$, the construction and penalty reduce the reported inverse sensitivity
relative to the original operator by approximately $372\times$ and $11\times$,
respectively.
\end{abstract}


\section{Introduction}
\label{sec:intro}

A data-driven turbulence closure supplies a term in the mean-flow equations
\citep{duraisamy2019}. Its substitution defines a new boundary-value problem: does a
solution exist, is it unique, and how sensitive is it to modelling error?

The conditioning of the closed equations has been analysed. \citet{wu2019} showed that
Reynolds-averaged Navier--Stokes (RANS) equations with an explicit data-driven Reynolds
stress can be ill-conditioned; \citet{brener2021} found the equations ill conditioned in every case they
analysed, unless DNS mean-velocity information was supplied.
Both analyses hold the closure frozen, so the operator analysed is that of the unclosed
equations. A trained closure deployed in a solver is a function of the velocity field, its
derivative enters the Jacobian, and the problem whose well-posedness is in question is the
coupled one.

Energy-neutral closures preserve the state-energy budget; related energy-conserving
neural closures support stability of time-dependent evolution \citep{vangastelen2024}.
Energy neutrality alone does not imply tangent dissipation or steady well-posedness.
For the steady problem, we instead seek monotonicity with a uniform coercive margin,
which gives a sufficient guarantee under the assumptions of theorem \ref{thm:main}.

The distinction between secant and tangent diffusion also appears in non-Newtonian
rheology as apparent viscosity, $\tau/\dot\gamma$, versus differential viscosity,
$d\tau/d\dot\gamma$. Positive apparent viscosity can coexist with a negative slope
of the flow curve. In the Johnson--Segalman model, a non-monotone constitutive curve
admits multiple steady plane Poiseuille solutions in appropriate parameter regimes
\citep{johnson1977,fyrillas1999}. This provides a constitutive analogue for neural
closures: a positive eddy viscosity need not imply a monotone tangent response or a unique steady solution.

We study the coupled mean-flow--closure system through its full tangent Jacobian
and prove two claims. For a strongly monotone original operator, global tangent dissipation of
the closure guarantees well-posedness and bounds inverse sensitivity, under the assumptions of
theorem \ref{thm:main}. For a general original operator, the closure improves or
preserves tangent dissipation without guaranteeing uniqueness. We then present two
ways to promote tangent dissipation: a partial construction enforcing non-negative
tangent diffusion, and a sampled penalty on tangent reaction. These complement
implicit treatment of the linear stress \citep{wu2018,wu2019,brener2021}, hybrid
eddy-viscosity treatment \citep{basara2003}, and non-negative eddy-viscosity constraints
\citep{mcconkey2022}; controlling the secant viscosity alone does not control the full
tangent operator. Related work in solid mechanics constructs polyconvex neural energies
supporting existence under additional growth assumptions \citep{tac2022}, and monotone
peridynamic neural models with conditional uniqueness under small deformations
\citep{wang2026}. Finally, channel-flow experiments assess conditioning, solver
convergence and accuracy across training seeds, while multiple initial conditions
test for distinct solutions of each fixed closure.
To our knowledge, this is the first sufficient existence-and-uniqueness result for a
fully coupled RANS problem with a state-dependent neural closure.

\section{Secant and tangent diffusion}
\label{sec:tangent}

Diffusion is one contribution to tangent dissipation; effective reaction can also
contribute. We first examine the diffusion contribution, distinguishing the coefficient
in the constitutive law from that governing its response to perturbations.

Let $\tau$ be the Reynolds shear stress and $S$ the mean shear, and let the closure be
written in eddy-viscosity form, $\tau=\nut(U,S)\,S$, for $S\ne0$. If $\tau(U,0)=0$ and
$\tau$ is differentiable in $S$, the quotient extends continuously to $S=0$ as
$\tau_S(U,0)$. This defines two associated diffusion coefficients.
The \emph{secant} diffusion, of
coefficient
\begin{equation}
\nut=\tau/S ,
\label{eq:secant}
\end{equation}
is the diffusion the constitutive relation asserts about the state, and it is the
diffusion the equation contains when the closure is written as $\D(\nut\D U)$; the
physical constraint of non-negative production, $P=\tau S\ge0$, is a constraint on it. The
\emph{tangent} diffusion, of coefficient
\begin{equation}
\nutan=\frac{\partial\tau}{\partial S}=\nut+S\,\frac{\partial\nut}{\partial S},
\label{eq:tangent}
\end{equation}
is the response of the closure to a perturbation of the state, and it is the diffusion the
Jacobian contains. On the stress--strain curve, $\nut\ge0$ says that the curve lies in the
first and third quadrants and $\nutan\ge0$ that it is monotone; with $\tau(0)=0$ the second
implies the first and not conversely.

Thus $\nut$ and $\nutan$ correspond to the apparent and differential viscosities in
non-Newtonian fluids. The latter enters the linearised diffusion operator
and hence sufficient conditions for well-posedness. A softplus on $\nut$ fixes
the sign of \eqref{eq:secant} and leaves \eqref{eq:tangent} free, which is negative
wherever $S\,\partial_S\nut<-\nut$. Implicit treatment writes $-\D(\nut\D\,\cdot)$
on the left of the equation and so places the secant diffusion in the operator
applied at each iteration, without ensuring uniqueness of the coupled equation
(\S\ref{sec:res:solve}). The observed multiplicity is analogous to that in
non-Newtonian flows: non-monotone constitutive curves in the Johnson--Segalman model
admit multiple steady Poiseuille solutions, with branch selection depending on the
flow history \citep{johnson1977,fyrillas1999}.

\section{Well-posedness of the coupled problem}
\label{sec:wellposed}

\subsection{The coupled residual and tangent dissipation}
\label{sec:coupled}
\label{sec:mechanisms}

Let $N$ collect the unclosed mean-flow terms and $F$ the learned closure. The closed problem is
\begin{equation}
R(U):=N(U)+F(U)=0,\qquad J(U)=A(U)+B(U),\quad A=N'(U),\quad B=F'(U).
\label{eq:coupled}
\end{equation}
The frozen analysis sets $B=0$.

Global tangent dissipation means $\langle v,B(U)v\rangle\le0$ for every state and
perturbation. For a scalar one-dimensional differential expression
$B=c_0+c_1\D+\D(c_2\D\,\cdot)$ on a bounded interval $I=(y_L,y_R)$,
with $\D=\partial_y$ and homogeneous Dirichlet perturbations $v\in H^1_0(I)$,
its symmetric quadratic form separates into reaction and diffusion:
\begin{equation}
\langle v,Bv\rangle=\int_I(c_0-\tfrac12c_1')v^2\,dy
-\int_I c_2(v')^2\,dy .
\label{eq:lemma-sym}
\end{equation}
Tangent dissipation includes both diffusion and reaction. Non-positive reaction and
non-negative diffusion suffice.
Continuum expressions below are understood as
forms on $H^1_0(I)$ paired with its dual, rather than bounded operators on $L^2$.

\subsection{Existence, uniqueness and inverse sensitivity}
\label{sec:theorem}

For the cell-integrated finite-volume residual, let $D=\operatorname{diag}(\dV_i)>0$ and set
\begin{equation}
x=D^{1/2}U,\qquad \mathcal R(x)=D^{-1/2}R(D^{-1/2}x),\qquad
\mathcal R'(x)=\tilde J(U):=D^{-1/2}J(U)D^{-1/2}.
\label{eq:weighted}
\end{equation}
Let $\tilde A=D^{-1/2}AD^{-1/2}$ and $\tilde B=D^{-1/2}BD^{-1/2}$. Write $\|v\|_D=\|D^{1/2}v\|_2$ and $\|r\|_{D^{-1}}=\|D^{-1/2}r\|_2$.

\begin{theorem}
\label{thm:main}
Suppose $N,F\in C^1(\mathbb R^n;\mathbb R^n)$ and the closure is globally tangent
dissipative, $\sym\tilde B(U)\preceq-bI$ for some $b\ge0$ at every state.

\textit{(i) Monotone baseline.} If $\sym\tilde A(U)\preceq-aI$ uniformly with
$a\ge0$ and $m:=a+b>0$, then $-\mathcal R$ is strongly monotone with constant $m$, $R(U)=0$ has exactly one
solution $U_*$, and $\sigma_{\min}(\tilde J(U))\ge m$ at every state. For any reference
field $\uref$, the following global bound holds without linearisation:
\begin{equation}
\|U_*-\uref\|_D\le \frac{\|R(\uref)\|_{D^{-1}}}{m}.
\label{eq:eJr}
\end{equation}
More generally, the solutions of $R(U_i)=f_i$ satisfy
$\|U_1-U_2\|_D\le m^{-1}\|f_1-f_2\|_{D^{-1}}$.

\textit{(ii) General baseline.} Without monotonicity of the baseline, the closure still
satisfies $\sym\tilde J(U)\preceq\sym\tilde A(U)-bI$ at each state, hence
$\omega(J(U))\le\omega(A(U))-b$, where $\omega$ is defined in \eqref{eq:omega}.
This comparison alone implies neither existence nor uniqueness.
\end{theorem}

\begin{proof}
For (i), adding the two symmetric-form bounds gives $\sym\tilde J\preceq-(a+b)I$.
Integrating the Jacobian along the segment from $z$ to $x$ gives
$\langle-\mathcal R(x)+\mathcal R(z),x-z\rangle\ge m\|x-z\|_2^2$.
This yields coercivity and uniqueness; continuity and finite-dimensional coercivity give
existence. Cauchy--Schwarz gives the singular-value and residual bounds. For (ii),
add $\sym\tilde B\preceq-bI$ to $\sym\tilde A$ and take the largest eigenvalue. The full
proof and a continuum counterpart with explicit function-space assumptions are in \S\ref{sec:deriv}.
\end{proof}

Part (i) guarantees that the closed flow problem has a solution and that this solution
is unique. It also bounds the a posteriori velocity error by the a priori momentum
imbalance $\|R(\uref)\|_{D^{-1}}$ divided by the dissipation margin $m$.
This imbalance measures the closure fitting error after insertion into the flow equations,
including any discretisation error, rather than the training loss alone.
Physically, $m$ lower-bounds the weakest damping of a velocity perturbation: a larger
margin limits how strongly a given closure error can alter the predicted flow.
With this residual sign convention, monotonicity refers to $-N$; the viscous
operator $N(U)=1+\nu U^{\prime\prime}$ has this property.
Mere monotonicity without a positive total margin does not suffice for (i).
Both parts require dissipation of the full closure tangent, including diffusion and reaction contributions.

\subsection{The measured index and its limits}
\label{sec:index}

On the assembled Jacobian define
\begin{equation}
\omega(J):=\lambda_{\max}(\sym\tilde J),\qquad
\sigma_{\min}(\tilde J)\ge-\omega(J)\quad\text{if }\omega(J)<0.
\label{eq:omega}
\end{equation}
The weighting matches the continuum $L^2$ metric (\S\ref{sec:fv}). The inverse sensitivity is
$1/\sigma_{\min}(\tilde J)$; its ratio to the frozen value is
$\amp:=\sigma_{\min}(\tilde A)/\sigma_{\min}(\tilde J)$, so $\amp<1$ indicates reduced
sensitivity. This is not the relative condition number
$\sigma_{\max}(\tilde J)/\sigma_{\min}(\tilde J)$.

We measure $J$ at $\uref$. A negative $\omega$ there controls the local inverse but does
not verify the uniform hypothesis of theorem \ref{thm:main}(i). If $\tilde J(\uref)$ is
invertible and the transformed Jacobian is locally Lipschitz, linearisation gives, with
$\hat e=D^{1/2}(U_*-\uref)$ and $\hat r=D^{-1/2}R(\uref)$,
\begin{equation}
\hat e=-\tilde J(\uref)^{-1}\hat r+O(\|\hat e\|_2^2).
\label{eq:local}
\end{equation}
Consequently $\|\hat r\|_2/\sigma_{\min}(\tilde J(\uref))$ bounds the linearised error,
not in general the nonlinear error. Unlike \eqref{eq:eJr}, it uses information at one
state only. Distinct roots imply $\omega(J)\ge0$ somewhere on their connecting segment;
strict positivity need not follow (\S\ref{sec:deriv}).

\section{Channel flow and the closure forms}
\label{sec:channel}

\subsection{Mean-momentum equation and reference data}
\label{sec:equation}

Fully developed, statistically steady plane channel flow reduces the streamwise mean
momentum equation to an ordinary differential equation in the wall-normal coordinate $y$.
Non-dimensionalising with $u_\tau=h=\rho=1$ gives $\nu=1/\Ret$ and unit forcing, and with
$\tau:=-\langle u'v'\rangle$,
\begin{equation}
\begin{gathered}
N(U)+F(U)=0,\qquad N(U):=1+\nu\,\DD U,\qquad F:=\D\tau,\\
y\in[0,2],\qquad U(0)=U(2)=0 .
\end{gathered}
\label{eq:pde}
\end{equation}
so that $A=\D(\nu\D\,\cdot)$ is symmetric negative definite and the network is the only
nonlinearity. Features are restricted to $(U,S,\Ret)$ with $S:=\D U$.
Data are the DNS of \citet{leemoser2015} at $\Ret\in\{180,550,1000,2000,5200\}$ on
wall-clustered meshes of $n=36$ to $1040$ cells (\S\ref{sec:fv}).

\subsection{The three closure forms}
\label{sec:forms}

All forms use the inputs $(U,S,\Ret)$ and the same hidden-layer width and depth,
with the output heads specified below. Differentiating $F$ gives $B$, whose symmetric
part is grouped using \eqref{eq:lemma-sym}.

\emph{Form A, force.} $F=\mathcal{N}^{A}$, so $B=F_U+F_S\D$ with $F_U:=\partial F/\partial
U$, $F_S:=\partial F/\partial S$, and
\begin{equation}
\sym B=F_U-\tfrac12\,\D F_S ,
\label{eq:symA}
\end{equation}
a pure reaction: Form A contributes no diffusion.

\emph{Form B, stress.} $F=\D\mathcal{N}^{B}$, $\tau=\mathcal{N}^{B}$. With
$\tau_U:=\partial\tau/\partial U$ and $\nutan=\partial\tau/\partial S$, in conservative form
\begin{equation}
B=\underbrace{\D\big(\nutan\,\D\;\cdot\big)}_{B_S}+\underbrace{\D\big(\tau_U\,\cdot\big)}_{B_U},
\qquad
\sym B=\D\big(\nutan\,\D\;\cdot\big)+\tfrac12\,\partial_y\tau_U ,
\label{eq:symB}
\end{equation}
the tangent diffusion and a reaction of coefficient $\tfrac12\partial_y\tau_U$.

\emph{Form C, eddy viscosity and residual.} $\tau=\nut S$, $F=\D\tau+F_{\mathrm{res}}$,
with $(\nut,F_{\mathrm{res}})=\mathcal{N}^{C}$. Its tangent combines the stress contribution
of Form B with the force contribution of Form A, where $\tau_U=S\,\partial_U\nut$.

\subsection{Two ways to constrain tangent dissipation}
\label{sec:enforce}

We follow two principles for promoting tangent dissipation. A constraint imposed
\emph{by construction} should be physically meaningful and expressible using only
local flow quantities, so it can be embedded in the constitutive law. A \emph{penalty}
can address more general conditions, including spatial derivatives of tangent
coefficients, but acts only at sampled states and provides no global guarantee.
The following two closures realise these principles.

\emph{Realisation by construction: tangent diffusion (B1).} Following \citet{wehenkel2019},
we parameterise a non-negative tangent viscosity and integrate it to obtain the stress:
\begin{equation}
\tau=\int_0^S\nutan(U,s,\Ret;\theta)\,ds,
\qquad \nutan=\alpha\,\mathrm{softplus}(\mathcal N^B)\ge0,
\quad \alpha>0.
\label{eq:integralform}
\end{equation}
For the exact integral, $\tau(U,0)=0$, $\partial_S\tau=\nutan\ge0$ and
$\sym B_S\preceq0$, so the total tangent diffusion is at least $\nu$.
This constraint is partial: the reaction coefficient $\tfrac12\partial_y\tau_U$
remains free. The implementation uses eight-node quadrature; its relation to the
exact-integral guarantee is detailed in \S\ref{sec:training}.

\emph{Realisation by penalty: tangent reaction (A1).} Form A has no diffusion contribution.
We penalise the positive part of its effective reaction coefficient:
\begin{equation}
\mathcal L_\omega=\big\langle[F_U-\tfrac12\D F_S]_+^2\big\rangle.
\label{eq:penalty}
\end{equation}
The reference set includes the DNS states at all five Reynolds numbers, augmented
with $(U,S)=(\alpha\uref,\alpha S_{\rm ref})$ for
$\alpha\in\{0,1/4,1/2,3/4\}$. Added states carry only the penalty; regression targets
remain those of the DNS.
The penalty is evaluated directly on the reference profiles, without solving the
flow equations, and encourages tangent dissipation only at the sampled states.

For comparison, C0 constrains only the secant viscosity $\nut\ge0$ by softplus,
leaving its tangent sign unrestricted.

\section{Results}
\label{sec:results}

\subsection{Closures, training and a priori accuracy}
\label{sec:res:apriori}

Five closures are compared (table \ref{tab:closures}): the three forms of
\S\ref{sec:forms} and two with constraints in
\S\ref{sec:enforce}. Every closure uses the same two-hidden-layer perceptron of width 48
and the same inputs, $U/30$, the viscous stress $\nu\,\D U$ and a normalised $\log\Ret$.
Training minimises a balanced loss on $d\tau/dy$ by Adam followed by L-BFGS,
with a first-order optimality target and recorded stopping criteria (\S\ref{sec:training}).
Ten seeds per closure give 250 (closure, $\Ret$, seed) cases.

Table \ref{tab:closures} reports the a priori fitting error in the Reynolds force
$F=\partial_y\tau$ over ten seeds. Figure \ref{fig:apriori} shows representative
stress and force profiles, evaluated on DNS features without a flow solve. This
fitting error differs from the deployed residual $R(\uref)$, which uses the
finite-volume operator (\S\ref{sec:fv}).

\begin{table}
\centering\footnotesize\setlength{\tabcolsep}{3pt}
\begin{tabular}{@{}lllrrrrr@{}}
\toprule
closure & predicts & constraint on the operator & 180 & 550 & 1000 & 2000 & 5200 \\
\midrule
A0  force & $F = \mathcal{N}$ & none & 1.01 & 0.93 & 0.85 & 0.64 & 0.55 \\
A1  force $+$ penalty & $F = \mathcal{N}$ & penalty: $F_U-\tfrac12\D F_S\le0$ & 1.77 & 1.51 & 1.34 & 1.03 & 0.78 \\
B0  stress & $\tau = \mathcal{N}$ & none & 0.40 & 0.36 & 0.28 & 0.28 & 0.29 \\
B1  stress, $\int\!\nu_{\rm tan}$ & $\tau = \int_0^S \nutan\,ds$ & construction: $\nutan\ge0$ & 0.57 & 0.71 & 0.82 & 0.62 & 0.32 \\
C0  $\nu_t$ $+$ residual & $F = \partial_y(\nut\partial_yU)+F_{\rm res}$ & $\nut\ge0$: state viscosity only & 0.16 & 0.12 & 0.11 & 0.10 & 0.06 \\
\bottomrule
\end{tabular}


\caption{The five closures, and the a priori error in $F$ [\%] at each $\Ret$ (median over
ten seeds). $\mathcal{N}$ is the network, $S=\D U$, $F_{\rm res}$ the learned residual.}
\label{tab:closures}
\end{table}

\begin{figure}
\centerline{\includegraphics[width=\textwidth]{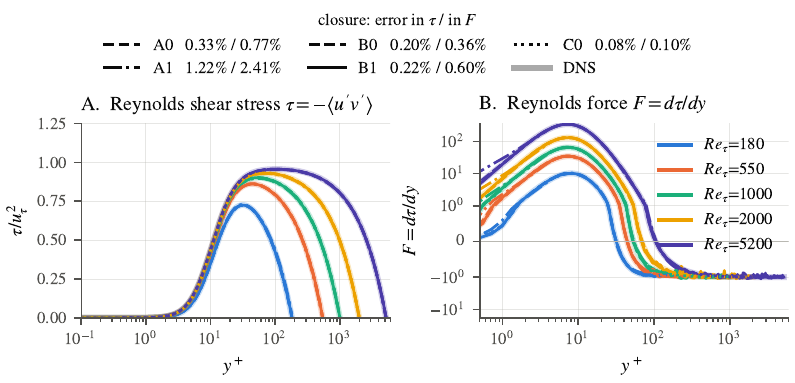}}
\caption{A priori test: DNS features in, no solve. Left, the Reynolds shear stress
$\tau=-\langle u'v'\rangle$; right, the Reynolds force $F=d\tau/dy$, the fitted quantity.
One trained network per closure is shown. Colour is $\Ret$, line style is closure,
and pale curves are DNS.}
\label{fig:apriori}
\end{figure}

\subsection{Dissipation and conditioning}
\label{sec:res:indices}

Figure \ref{fig:indices} gives $\omega(J)$ and $\amp$ at the DNS reference state;
ensemble means and standard deviations are tabulated in \S\ref{sec:bound}. B1 is dissipative
in 49 of 50 cases, A1 in all 50, C0 in 26, and A0 and B0 in none.
B0's $\omega$ reaches $3\times10^4$ at $\Ret=5200$.
B1's mean $\amp$ decreases from $0.075$ at $\Ret=180$ to $0.0027$ at $5200$.
The reciprocals of these mean ratios are $13$ and $372$: while the frozen operator's
inverse sensitivity grows by a factor of 28, B1's varies little.
The added tangent diffusion is consistent with this behaviour, although reaction
also contributes to the full Jacobian. At $\Ret=5200$, the mean ratios are $0.091$
for A1, $0.25$ for C0 and $10\pm12$ for A0. These are local sensitivity measurements,
not the relative condition number or a verification of the global theorem.

\begin{figure}
\centerline{\includegraphics[width=\textwidth]{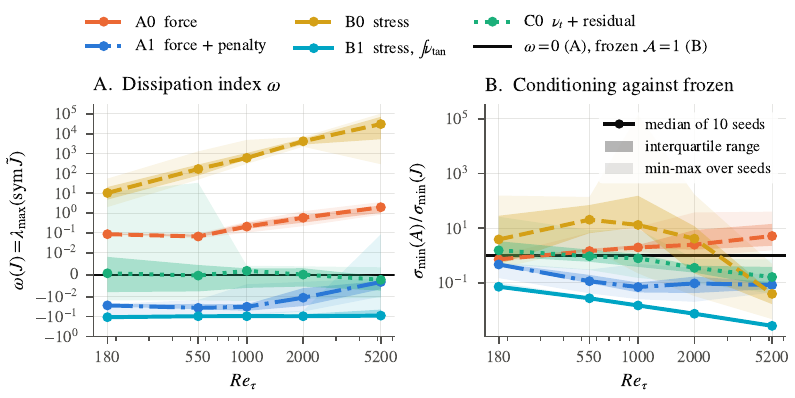}}
\caption{The two indices against $\Ret$; larger is worse in both. (\textit{a}) $\omega(J)$,
dissipative below the line. (\textit{b}) $\amp$, better than frozen below it. Median,
interquartile band and min--max envelope over ten seeds.}
\label{fig:indices}
\end{figure}

The secant and tangent viscosities of C0 are evaluated on all ten trained networks
using the DNS reference velocity and shear at the cell faces. The secant viscosity is non-negative at every face, seed and Reynolds number, with a minimum
of $+1.5\times10^{-70}$; the tangent viscosity \eqref{eq:tangent} is negative, and the
closure diffusion contribution negative, for five seeds of ten at $\Ret=180$, two at $550$ and
$1000$, and one at $2000$ and $5200$, reaching $-0.35\nu$, always in the upper half of the
channel, where $S<0$. The total coefficient $\nu+\nutan$ remains positive in these
measurements; negative closure diffusion alone is not loss of ellipticity.

\subsection{A posteriori accuracy and solution multiplicity}
\label{sec:res:solve}

For the ten-seed ensemble, every solve starts by damped Newton from $U\equiv0$ or by
Reynolds-number continuation ascending from $U\equiv0$ at $\Ret=180$, with the
preconditioner $P=A$; a closure is credited if either converges.

We measure velocity error as
$\epsU=\|U_*-\uref\|_D/\|\uref\|_D$ and retain the smaller error when both starts
converge. A converged solve is classified \emph{physical} if $\epsU<5\,\%$ and \emph{spurious}
otherwise; these labels refer only to velocity accuracy, not to a test of all physical
constraints. Reported velocity-error statistics use the accurate subset.

Table \ref{tab:useful} separates solver convergence from velocity accuracy.
A0 and B0 converge in 39 and 15 of the 50 cases, respectively, but none of their
solutions meets the accuracy threshold. C0 converges in 13 cases, with only one
accurate solution at each Reynolds number; the other eight are spurious.
Thus a converged residual alone does not imply an accurate mean flow.
In comparison, A1 and B1 both converge to accurate solutions for all ten seeds at
every Reynolds number. B1 has smaller mean errors and less variation across seeds:
its mean error decreases from $0.50\,\%$ at $\Ret=180$ to $0.02\,\%$ at $5200$,
whereas A1's means range from $0.81\,\%$ to $1.87\,\%$.
The largest individual errors are $4.478\,\%$ for A1 ($\Ret=5200$) and
$0.535\,\%$ for B1 ($\Ret=180$). Both methods for promoting tangent dissipation
therefore yield reliable convergence and accurate predictions across this ensemble,
with the diffusion construction giving the more accurate and reproducible results.

\begin{table}
\centering\footnotesize\setlength{\tabcolsep}{1.5pt}
\begin{tabular}{@{}lcccccrrrrr@{}}
\toprule
& \multicolumn{5}{c}{physical / converged, of 10} & \multicolumn{5}{c}{$\epsU$ [\%]} \\
\cmidrule(lr){2-6}\cmidrule(l){7-11}
closure & 180 & 550 & 1000 & 2000 & 5200 & 180 & 550 & 1000 & 2000 & 5200 \\
\midrule
A0 & 0\,/\,8 & 0\,/\,9 & 0\,/\,8 & 0\,/\,8 & 0\,/\,6 & -- & -- & -- & -- & -- \\
A1 & \textbf{10\,/\,10} & \textbf{10\,/\,10} & \textbf{10\,/\,10} & \textbf{10\,/\,10} & \textbf{10\,/\,10} & 1.02 $\pm$ 0.92 & 0.85 $\pm$ 0.87 & 0.81 $\pm$ 0.77 & 1.17 $\pm$ 1.07 & 1.87 $\pm$ 1.49 \\
B0 & 0\,/\,3 & 0\,/\,3 & 0\,/\,3 & 0\,/\,3 & 0\,/\,3 & -- & -- & -- & -- & -- \\
B1 & \textbf{10\,/\,10} & \textbf{10\,/\,10} & \textbf{10\,/\,10} & \textbf{10\,/\,10} & \textbf{10\,/\,10} & 0.50 $\pm$ 0.01 & 0.07 $\pm$ 0.00 & 0.06 $\pm$ 0.01 & 0.04 $\pm$ 0.01 & 0.02 $\pm$ 0.00 \\
C0 & 1\,/\,1 & 1\,/\,2 & 1\,/\,3 & 1\,/\,4 & 1\,/\,3 & 0.26 & 0.92 & 0.48 & 1.73 & 0.85 \\
\bottomrule
\end{tabular}


\caption{Left, \emph{physical\,/\,converged} of ten seeds; right, $\epsU$ at the physical
solutions, mean $\pm$ s.d. Bold: every seed physical.}
\label{tab:useful}
\end{table}

The ten-seed ensemble measures sensitivity to training initialization. To test
multiplicity of a fixed closed equation, we separately freeze one seed per
closure at $\Ret=1000$ and apply the same 10,000 initial profiles to each.
The initial condition set contains zero and laminar flow, plus 9,998 random profiles
\begin{equation}
U_0(y)=a\,y(2-y)+b\,\frac{w(y)}{\max_i|w(y_i)|},
\qquad w(y)=\sum_{j=1}^{16}\frac{\xi_j}{j^p}\sin\!\left(\frac{j\upi y}{2}\right),
\label{eq:random-ic}
\end{equation}
where $a\,y(2-y)$ is the parabolic base ($a=0$ omits it),
$\xi_j\sim\mathcal N(0,1)$ are independent, and $p=2$ or $1$ gives smoother or
more oscillatory profiles. The sine series is normalised at the cell centres $y_i$;
random amplitudes $a,b$ include reversed and large-amplitude flows (\S\ref{sec:supp:multistart}).
All profiles satisfy the no-slip wall conditions.
Figure \ref{fig:multimodal} shows 7, 1, 93, 1 and 4 distinct roots for A0, A1, B0, B1
and C0, respectively. Each representative has relative weighted residual below
$10^{-8}$ and relative Newton correction below $10^{-7}$; profiles separated by less
than $10^{-4}$ in relative weighted velocity distance are grouped together.
These results establish numerical multiplicity of the discrete equations for A0, B0
and C0. The single roots found for A1 and B1 provide empirical evidence of robustness
over the sampled starts, not a proof of uniqueness. The initialization distribution,
stopping rules and convergence counts are given in \S\ref{sec:supp:multistart}.

\begin{figure}
\centerline{\includegraphics[width=\textwidth]{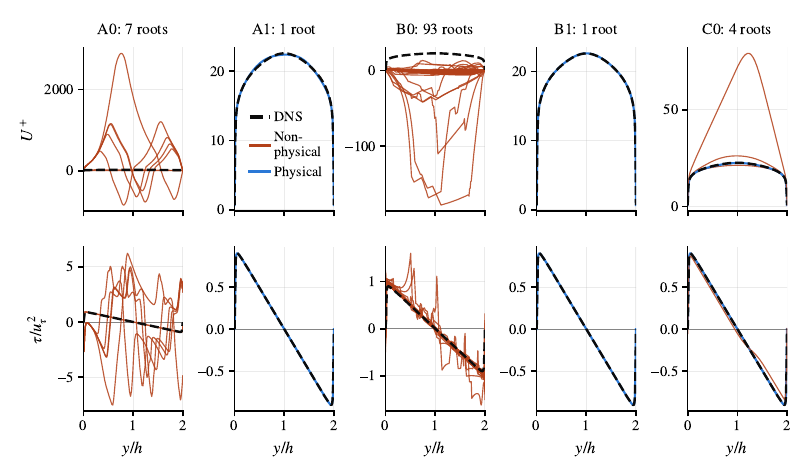}}
\caption{Distinct solutions of fixed closures at $\Ret=1000$, from 10,000 common initial
profiles per closure. Top, velocity; bottom, the reference-normalised stress diagnostic
$\tau_{\rm rec}=(1-y)-\nu S$, using the same discrete gradient for solutions and DNS
(\S\ref{sec:supp:multistart}).
The full channel is shown, with independent vertical ranges retaining all roots.
DNS is dashed; blue denotes the reference-like branch and brown the other roots.
For C0, the secondary near-DNS root ($\epsU=3.90\,\%$) is brown, while the closest root
($0.476\,\%$) is blue; this branch colouring differs from the $5\,\%$ accuracy
classification used in table \ref{tab:useful}.}
\label{fig:multimodal}
\end{figure}

The supplementary measurements distinguish solver efficiency from the sensitivity
of the closed flow problem. For a fixed residual, implicit treatment and preconditioning
change the iteration but not the root set. In the tested converged cases, frozen-viscous,
diffusion-augmented and full-Jacobian preconditioners give velocity errors agreeing to
$1.4\times10^{-12}$, while GMRES counts differ by an order of magnitude (\S\ref{sec:precond}).
Thus they change the work required to compute a solution, without correcting the
closure's modelling error. The Jacobian-based estimate has a different role:
$1/\sigma_{\min}(\tilde J)$ bounds the worst-case linear amplification, attained when
the residual excites the most sensitive direction. An actual closure error need not
align with that direction, so its amplification can be much smaller. Empirically,
the linearised upper estimate exceeds the realised error by a factor of $105\pm201$ over the
105 accurate solutions (\S\ref{sec:bound}). This illustrates its conservatism, without establishing
a global nonlinear bound or verifying the uniform-margin hypothesis of \eqref{eq:eJr}.

\section{Conclusions}
\label{sec:conclusions}

The coupled problem $R=N+F=0$ is controlled by the tangent Jacobian $J=N'+F'$.
For a monotone original operator, global tangent dissipation of the closure guarantees
well-posedness when the PDE or closure supplies a uniform coercive margin, under the
stated regularity assumptions. The margin also bounds inverse sensitivity. For a general
original operator, the closure adds tangent dissipation at each state, but this alone
guarantees neither existence nor uniqueness. Diffusion and reaction both contribute to
dissipation; controlling diffusion alone leaves the reaction unconstrained.

The partial diffusion construction and sampled reaction penalty find accurate
solutions in all 50 channel cases each, with maximum velocity errors of $0.535\,\%$
and $4.478\,\%$, respectively. At $\Ret=1000$, the 10,000-start test finds one root
for each constrained closure and multiple roots for A0, B0 and C0, with one fixed
network per closure. The reciprocals of the mean inverse-sensitivity ratios are
approximately $372$ and $11$ for the construction and penalty at $\Ret=5200$.
Neither approach certifies full global dissipation: the construction leaves reaction
free and uses quadrature, while the penalty is sampled. These experiments support
both approaches without establishing global uniqueness.

\backsection[Supplementary material]{\label{SupMat}Supplementary material (theorem
proofs and continuum counterpart, finite-volume realisation, training and constraint
implementation, solver comparison, local error estimates and multistart protocol) is appended after the references in this preprint.}

\backsection[Acknowledgements]{The authors used OpenAI's ChatGPT and Codex to assist
with drafting and language editing of the manuscript. The authors reviewed and edited
the resulting text and take full responsibility for the content of this work.}

\backsection[Funding]{This research was supported by the Defense Advanced Research
Projects Agency (DARPA) under the Automated Prediction Aided by Quantized Simulators
(APAQuS) program, Grant No. HR00112490526.}

\backsection[Declaration of interests]{The authors report no conflict of interest.}

\backsection[Data and code availability]{The DNS data are those of \citet{leemoser2015}.
The code is available at \url{https://github.com/zhangzhen117/coupled_PDE_NN_closure}.}

\bibliographystyle{jfm}
\bibliography{refs}

\clearpage
\setcounter{section}{0}
\setcounter{subsection}{0}
\setcounter{equation}{0}
\setcounter{figure}{0}
\setcounter{table}{0}
\renewcommand{\thesection}{S\arabic{section}}
\renewcommand{\thefigure}{S\arabic{figure}}
\renewcommand{\thetable}{S\arabic{table}}
\renewcommand{\theequation}{S\arabic{section}.\arabic{equation}}
\renewcommand{\theHsection}{supp.\arabic{section}}
\renewcommand{\theHsubsection}{\theHsection.\arabic{subsection}}
\renewcommand{\theHequation}{supp.\arabic{section}.\arabic{equation}}
\renewcommand{\theHfigure}{supp.\arabic{figure}}
\renewcommand{\theHtable}{supp.\arabic{table}}
\markboth{Z. Zhang and G. E. Karniadakis}{Supplementary material: neural turbulence closures}
\phantomsection
\pdfbookmark[0]{Supplementary material}{supplement}
{\raggedright\Large\bfseries Supplementary material\par}
\medskip
{\raggedright\large Well-posedness of neural turbulence closures and tangent dissipation\par}
\medskip

This supplementary material accompanies the article named above.
The sections follow the main paper's progression: tangent-form identities and
well-posedness proofs (\S\ref{sec:deriv}); discretisation, reference features and norms
(\S\ref{sec:fv}); training and the two constraint implementations (\S\ref{sec:training});
nonlinear solves and preconditioners (\S\ref{sec:precond}); local sensitivity estimates
and measured amplification (\S\ref{sec:bound}); and the fixed-checkpoint multistart
experiment (\S\ref{sec:supp:multistart}). Main-text equations, tables and figures are
identified explicitly; supplementary items carry the prefix S.

\section{Tangent forms and well-posedness proofs}
\label{sec:deriv}

\subsection{Symmetric quadratic form and coefficient bounds}

Let $I=(y_L,y_R)$ be a bounded interval, $V=H^1_0(I)$, and let the real coefficients satisfy
$c_0,c_2\in L^\infty(I)$ and $c_1\in W^{1,\infty}(I)$.
The expression $B=c_0+c_1\D+\D(c_2\D\,\cdot)$ defines a bounded map $V\to V^*$.
Integration by parts in the dual pairing gives
\begin{equation}
\langle Bv,v\rangle_{V^*,V}
=\int_I qv^2\,dy-\int_I c_2|v'|^2\,dy,
\qquad q=c_0-\tfrac12c_1'.
\label{eq:symform}
\end{equation}
Indeed $\int c_1v'v=-\tfrac12\int c_1'v^2$; the trace of $v$ vanishes at both walls.
Polarisation identifies the symmetric form with $q+\D(c_2\D\,\cdot)$.
This form identity does not assert that $\D$ is a skew-adjoint operator on the
Dirichlet domain in $L^2$.

If $c_2\ge\kappa\ge0$ and $q\le\beta$, then Poincar\'e's inequality gives
\begin{equation}
\omega_{L^2}(B):=\sup_{v\in V\setminus\{0\}}
\frac{\langle Bv,v\rangle}{\|v\|_{L^2}^2}
\le\beta-\kappa\lambda_1,\qquad \lambda_1=(\upi/|I|)^2.
\label{eq:coefficientbound}
\end{equation}
Thus $q\le0$ and $c_2\ge0$ suffice for dissipation. For a state-dependent Jacobian,
any claimed global margin requires these bounds uniformly over states, not only on a
reference profile. For the stress form the total diffusion is $\nu+\nutan$ and the
reaction is $q=\tfrac12\partial_y\tau_U$. If uniformly
$\nu+\nutan\ge\kappa>0$ and $q\le\beta$ with $\beta\ge0$ and
$\beta<\kappa\lambda_1$, then
\begin{equation}
-\langle J(U)v,v\rangle\ge
(\kappa-\beta/\lambda_1)\|v'\|_{L^2}^2.
\label{eq:Vmarg}
\end{equation}
The non-negative shear branch of B1 alone does not supply the required reaction bound.

\subsection{Proof of the discrete theorem}

For part (i), use the transformed residual $\mathcal R$ defined in main-text \textup{(\mainref{eq:weighted})} and set
$T=-\mathcal R$. The separate baseline and closure hypotheses imply
$\sym\mathcal R'=\sym\tilde A+\sym\tilde B\preceq-(a+b)I=-mI$.
For $w=x-z$, the fundamental theorem of calculus yields
\begin{equation}
\langle T(x)-T(z),w\rangle
=-\int_0^1w^\top\mathcal R'(z+tw)w\,dt\ge m\|w\|_2^2.
\label{eq:discretemonotone}
\end{equation}
In particular $\langle T(x),x\rangle\ge m\|x\|_2^2-\|T(0)\|_2\|x\|_2$.
For any $g\in\mathbb R^n$, the continuous map $H=T-g$ points strictly outwards
on the sphere of radius $r>(\|T(0)\|_2+\|g\|_2)/m$. It has a zero in the ball:
otherwise $x\mapsto-rH(x)/\|H(x)\|_2$ maps the ball continuously to itself and
Brouwer's fixed-point theorem would give a boundary point with
$\langle H(x),x\rangle<0$, a contradiction. Thus $T$ is onto; strong monotonicity
makes it one-to-one. This proves existence and uniqueness for every forcing
\citep[for finite-dimensional nonlinear equations]{ortegarheinboldt1970}.

For $M=\mathcal R'(x)$, Cauchy--Schwarz implies
$\|Mv\|_2\|v\|_2\ge|v^\top Mv|\ge m\|v\|_2^2$; hence
$\sigma_{\min}(M)\ge m$ and $M$ is invertible. Applying
\eqref{eq:discretemonotone} to a root $x_*$ and an arbitrary $x$ gives
$m\|x_*-x\|_2\le\|\mathcal R(x)\|_2$. Applying it to two forced solutions gives
$m\|x_1-x_2\|_2\le\|g_1-g_2\|_2$. Transforming back proves the residual and
forcing bounds in the main theorem. These estimates have no Taylor remainder.
The constant $1/m$ is sharp: equality holds for $\mathcal R(x)=-mx+c$.
A nonlinear example is $\mathcal R_i(x)=c_i-mx_i-x_i^3$, whose symmetric
Jacobian is $-\operatorname{diag}(m+3x_i^2)$. The global bound remains valid,
although a local inverse at a distant reference need not bound the nonlinear error:
for $\mathcal R(x)=-x-x^3$, $x_*=0$ and $x_0=1$, the error is $1$ but the
linearised estimate is $|\mathcal R(1)/\mathcal R'(1)|=1/2$.

For part (ii), at any fixed state and for every $v$,
\begin{equation}
v^\top\sym\tilde Jv\le v^\top\sym\tilde Av-b\|v\|_2^2.
\end{equation}
Taking the supremum over unit vectors proves $\omega(J)\le\omega(A)-b$.
No baseline monotonicity is needed for this comparison. However, $N(x)=x$ and
$F(x)=-x$ give $R(x)=0$ for every $x$: the closure has margin $b=1$ and reduces
the numerical abscissa from $1$ to $0$, yet uniqueness fails. With $N(x)=x+1$
and the same closure, $R(x)=1$ has no root. Thus the dissipation comparison supplies
neither conclusion in general, even when the baseline itself is invertible.
It does not imply improved smallest singular values either: these scalar examples
reduce that value from $1$ to $0$. If the sum independently has a uniform negative
margin, the proof of (i) applies, but this is additional information.

For a monotone baseline, at least one positive uniform margin is needed for the
stated global well-posedness and inverse bounds. Taking $N=F=0$ shows why two
merely monotone maps are insufficient. For $a\ge0$ and $b>0$ the inverse bound is
at most $1/b$, improved to $1/(a+b)$ when $a>0$. Negative eigenvalues alone do not
imply the symmetric-form hypothesis for a nonnormal Jacobian; and inverse sensitivity
is not the relative condition number, which also involves $\sigma_{\max}(\tilde J)$.

\subsection{A continuum theorem}

Equip $V=H^1_0(I)$ with $\|v\|_V=\|v'\|_{L^2}$ and use the induced dual norm on
$V^*=H^{-1}(I)$. Suppose $T=-R:V\to V^*$ is bounded on bounded sets,
hemicontinuous (the map $t\mapsto\langle T(u+tv),w\rangle$ is continuous), and
uniformly strongly monotone: for some $\alpha>0$,
\begin{equation}
\langle T(u)-T(v),u-v\rangle\ge\alpha\|u-v\|_V^2
\quad\text{for all }u,v\in V.
\label{eq:continuummonotone}
\end{equation}
Then for every $g\in V^*$ there is exactly one $u\in V$ with $T(u)=g$.
To prove existence, observe that
\begin{equation}
\frac{\langle T(u),u\rangle}{\|u\|_V}
\ge\alpha\|u\|_V-\|T(0)\|_{V^*}\longrightarrow\infty.
\end{equation}
The Browder--Minty surjectivity theorem on the reflexive space $V$ applies
\citep{zeidler1990}; strong monotonicity gives uniqueness. For the zero-forcing
root and any reference $v\in V$, dual Cauchy--Schwarz gives
\begin{equation}
\|u_*-v\|_V\le\alpha^{-1}\|R(v)\|_{V^*},\qquad
\|u_*-v\|_{L^2}\le\frac{\|R(v)\|_{V^*}}{\alpha\sqrt{\lambda_1}}.
\label{eq:continuumbound}
\end{equation}
If $R(v)$ is represented by an $L^2$ function, then
$\|R(v)\|_{V^*}\le\lambda_1^{-1/2}\|R(v)\|_{L^2}$ and consequently
$\|u_*-v\|_{L^2}\le(\alpha\lambda_1)^{-1}\|R(v)\|_{L^2}$.
The same proof bounds differences of solutions by differences of forcing in $V^*$.
No continuum singular-value assertion is needed.

To express the continuum hypothesis in terms of the original PDE and the added closure,
suppose $N,F:V\to V^*$ are bounded on bounded sets and hemicontinuous, and for every
$u,v\in V$, with $w=u-v$,
\begin{equation}
\langle-N(u)+N(v),w\rangle\ge a_V\|w\|_V^2,\qquad
\langle-F(u)+F(v),w\rangle\ge b_V\|w\|_V^2,
\end{equation}
where $a_V,b_V\ge0$ and $a_V+b_V>0$. Adding these inequalities gives
\eqref{eq:continuummonotone} with $\alpha=a_V+b_V$. For continuously Fr\'echet
differentiable maps, the inequalities follow from the corresponding uniform tangent
form bounds by integration along segments. If $a_V\ge0$ and $b_V>0$, the inverse
from $V^*$ to $V$ has Lipschitz constant at most $1/b_V$; a positive baseline margin
improves it to $1/(a_V+b_V)$. These are energy-space margins; an $L^2$ tangent bound
alone does not supply the required $H^1_0$ coercivity.

For a general differentiable baseline, global tangent dissipation of $F$ gives only
the continuum form comparison
$\langle (N'(u)+F'(u))w,w\rangle\le\langle N'(u)w,w\rangle$ at each state.
Whenever the $L^2$ numerical abscissae of these forms are finite, the same Rayleigh
quotient argument gives $\omega_{L^2}(N'+F')\le\omega_{L^2}(N')$.
Without an additional coercive monotonicity hypothesis this comparison does not
establish existence or uniqueness. A partial diffusion constraint proves it for
$N'+B_S$, not for $N'+B_S+B_U$ unless the full reaction--diffusion form is controlled.

One sufficient channel setting is $R(u)=f+\nu u''+F(u)$ with $f\in V^*$ and
$F:V\to V^*$ continuously Fr\'echet differentiable and bounded on bounded sets.
If $\langle F'(u)w,w\rangle\le0$ for every $u,w\in V$, integration along segments
gives \eqref{eq:continuummonotone} with $\alpha=\nu$. More generally the uniform
form estimate \eqref{eq:Vmarg} gives $\alpha=\kappa-\beta/\lambda_1$ when the
residual has this mapping regularity. These are explicit sufficient assumptions;
smoothness of a neural network as a function of its finite-dimensional inputs alone
does not establish them on $H^1_0$. An elementary admissible example is a stress
$\tau(y,S)$, independent of $u$, measurable in $y$, continuous and monotone in $S$,
with $|\tau(y,S)|\le a(y)+C|S|$ for $a\in L^2(I)$. Its weak divergence defines a
bounded, hemicontinuous map $V\to V^*$ and $-F$ is monotone, so the viscous term
supplies $\alpha=\nu$. This example is not a global certificate for the trained B1,
which also depends on $U$.

\subsection{What multiplicity implies}

If distinct $x,z$ are roots of a $C^1$ discrete residual, then
$\int_0^1w^\top\mathcal R'(z+tw)w\,dt=0$, with $w=x-z\ne0$.
Continuity implies a point on the segment where the integrand is non-negative,
so $\lambda_{\max}(\sym\mathcal R')\ge0$ there. Strict positivity need not occur:
the scalar $C^1$ residual
\begin{equation}
\mathcal R(x)=\begin{cases}-x^3,&x<0,\\0,&0\le x\le1,\\-(x-1)^3,&x>1\end{cases}
\end{equation}
has an interval of roots while $\mathcal R'(x)\le0$ everywhere. Conversely, a
positive numerical abscissa at a sampled state does not itself prove multiplicity.

\subsection{Closure-specific derivations}

\emph{Force form.} For $F(U,S)$, $Bv=F_Uv+F_Sv'$. Thus $c_2=0$ and
$q=F_U-\tfrac12\partial_yF_S$, the coefficient used in the penalty of
main-text \textup{(\mainref{eq:penalty})}. There is no second-order diffusion contribution.

\emph{Conservative grouping of the stress form.} With $\tau_U=\partial\tau/\partial U$ and
$\nutan=\partial\tau/\partial S$, the chain rule gives $B=\D\circ(\tau_U+\nutan\D)$, which is
the symmetric-form identity with $c_2=\nutan$, $c_1=\tau_U$ and $c_0=\partial_y\tau_U$, hence
$\sym B=\D(\nutan\D\,\cdot)+\tfrac12\partial_y\tau_U$. Expanding $\D(\nutan\D\,\cdot)$ into
$\nutan\DD+(\D\nutan)\D$ splits one self-adjoint operator into a second-order and a
first-order piece, neither self-adjoint on its own; the term $(\D\nutan)\D$ is the gradient
of the diffusivity, not an advection, and its symmetric part $-\tfrac12\DD\nutan$ cancels
against that of $\nutan\DD$. The operator is of second order, not third, since
$\partial\tau/\partial S$ multiplies one derivative of $v$ and the outer $\D$ returns it to
second order.

\emph{The bound for the stress form.} By \eqref{eq:coefficientbound} with $c_2=\nu+\nutan$ and
$c_0-\tfrac12c_1'=\tfrac12\partial_y\tau_U$,
$\omega(J)\le-(\nu+\inf_y\nutan)\lambda_1+\tfrac12\sup_y(\partial_y\tau_U)_+$ whenever
$\nutan\ge0$, with $\lambda_1=(\upi/|I|)^2$ the first Dirichlet eigenvalue of $-\DD$ on
$I$.

\emph{Form C.} With $\tau=\nut(U,S)S$, $\partial\tau/\partial S=\nut+S\,\partial_S\nut$ and
$\partial\tau/\partial U=S\,\partial_U\nut$, so the first term of Form C is Form B with these
coefficients and the residual head contributes $(F_{\rm res})_U-\tfrac12\partial_y(F_{\rm
res})_S$, giving the reaction contribution for Form C in main-text \S\mainref{sec:forms}.

\section{Discretisation, reference features and norms}
\label{sec:fv}

\subsection{Mesh and closure evaluation}
The channel $y\in[0,2]$ is discretised by a cell-centred finite-volume mesh with
faces $y_{i\pm1/2}$, centres $y_i$, volumes $\dV_i=y_{i+1/2}-y_{i-1/2}$ and
centre-to-centre distances $d_{i+1/2}=y_{i+1}-y_i$. At a wall, $d$ is the distance
from the adjacent cell centre to the wall, where $U=0$. The five meshes have
$n=36$, 110, 200, 400 and 1040 cells at nominal $\Ret=180$, 550, 1000, 2000 and 5200.
The DNS velocity is interpolated onto the cell centres to give $\uref$.

The force form is evaluated at cell centres, with shear $S=GU$ obtained from a
three-point gradient stencil incorporating the wall values. Stress-form closures
are evaluated at faces, using linearly interpolated velocity and
$S_{i+1/2}=(U_{i+1}-U_i)/d_{i+1/2}$. Their stress and its derivatives with respect
to the unknown velocities are set to zero at the two walls. The hybrid uses these
face quantities for its stress branch and cell-centred quantities for its residual head.
All signs below use the paper convention $\tau=-\langle u'v'\rangle$; the code stores
the opposite stress sign and converts it when assembling the residual.

Let $(Eq)_i=q_{i+1/2}-q_{i-1/2}$ and $D=\operatorname{diag}(\dV_i)$. The integrated
closure contribution is $D F$ for the force form, $E\tau$ for the stress form,
and $E\tau+D F_{\rm res}$ for the hybrid. All residuals are cell-integrated.
The viscous matrix is $A=\mathcal D(\nu)$, where
\begin{equation}
[\mathcal D(\kappa)v]_i=
\frac{\kappa_{i+1/2}}{d_{i+1/2}}(v_{i+1}-v_i)
-\frac{\kappa_{i-1/2}}{d_{i-1/2}}(v_i-v_{i-1}),
\qquad v_0=v_{n+1}=0.
\label{eq:diffop}
\end{equation}
The integrated forcing is $s_i=\dV_i$, so the coupled residual is
$R(U)=s+AU+F_h(U)$, with $F_h$ the appropriate integrated closure contribution.

\subsection{Jacobian and weighted coordinates}
Automatic differentiation through the assembled stencils gives $J=A+B$.
For the force form,
$B=D\operatorname{diag}(F_U)+D\operatorname{diag}(F_S)G$.
For a face-stress Jacobian $T_f=\partial\tau_f/\partial U$, $B=ET_f$.
Differentiating the complete stencil includes interpolation, shear dependence and wall
conditions. The continuum reaction--diffusion identities guide the interpretation,
but their coefficients are not substituted for the discrete symmetric part.

The matrix $\mathcal D(\kappa)$ is symmetric on a non-uniform mesh and
\begin{equation}
v^\top\mathcal D(\kappa)v=-\sum_f\frac{\kappa_f}{d_f}(\Delta v)_f^2.
\end{equation}
Thus non-negative face coefficients yield a non-positive diffusion form. Removing
wall stress also removes its tangent wall contribution; molecular viscosity retains
the Dirichlet wall terms. Variable-coefficient first-order stencils need not reduce
to diagonal reaction matrices after symmetrisation, especially on a stretched mesh.

The field and integrated-residual norms are, respectively,
\begin{equation}
\|v\|_D=\|D^{1/2}v\|_2,\qquad
\|r\|_{D^{-1}}=\|D^{-1/2}r\|_2.
\end{equation}
The coordinates $x=D^{1/2}U$, $\mathcal R(x)=D^{-1/2}R(D^{-1/2}x)$ give
$\mathcal R'(x)=\tilde J=D^{-1/2}JD^{-1/2}$, as in main-text
\textup{(\mainref{eq:weighted})}. This matrix is similar to $D^{-1}J$, and its Euclidean
singular values measure sensitivity in the stated physical norms. All reported
$\omega(J)$ and $\sigma_{\min}(\tilde J)$ use this weighting.

\subsection{Fitting error versus deployed residual}
The a priori force error in main-text table \mainref{tab:closures} uses DNS-grid
velocity and shear and differentiation on that grid. The curves in main-text
figure \mainref{fig:apriori} use one representative network per closure (seed 0).
By contrast, $R(\uref)$ uses the interpolated velocity and the deployed finite-volume
stencils. Even a small regression error need not give the same-sized deployed residual.
The relative velocity error is
$\epsU=\|U_*-\uref\|_D/\|\uref\|_D$; the reference residual is
$\hat r=D^{-1/2}R(\uref)$. These distinct quantities must not be interchanged in the
inverse-sensitivity estimate. The pointwise momentum imbalance is $D^{-1}R$, whose
volume-weighted norm equals $\|R\|_{D^{-1}}$.

\section{Training and implementation of the constraints}
\label{sec:training}

\subsection{Regression loss and optimisation}
The networks use two hidden layers of width 48 and features $U/30$, $\nu S$ and a
normalised $\log\Ret$, with outputs specified in main-text \S\mainref{sec:forms}.
All five Reynolds numbers enter training. Let $F_i^{\rm ref}$ be the DNS force target
and $w_i$ the DNS-grid weights, normalised to give each Reynolds number equal weight
and $\sum_iw_i=1$. The balanced regression loss is
\begin{equation}
\mathcal L_{\rm fit}=
\frac{\sum_iw_i(F_i-F_i^{\rm ref})^2}{\sum_iw_i(F_i^{\rm ref})^2}
+\sum_iw_i\left(\frac{F_i-F_i^{\rm ref}}{|F_i^{\rm ref}|+1}\right)^2.
\end{equation}
The first term controls global relative error; the second prevents the large near-wall
force from dominating the smaller outer-layer values. Force targets and features are
formed on the DNS grid, rather than by solving the deployed flow equations.

Optimisation uses Adam followed by L-BFGS, with $\|\nabla\mathcal L\|_2\le10^{-5}$
as the first-order target. Stagnation and checkpoint-selection criteria can stop
optimisation before that target; the achieved gradient and stopping reason are retained
with each checkpoint. A1 starts from the trained A0 network of the same seed.
The final gradient norms across the reported ensemble are:
\begin{center}
\begin{tabular}{@{}lr@{}}
\toprule
Closure & Range of $\|\nabla\mathcal L\|_2$\\
\midrule
A0 & $4.2\times10^{-4}$--$5.1\times10^{-3}$\\
A1 & $6.4\times10^{-3}$--$5.2\times10^{-1}$\\
B0 & $7.0\times10^{-4}$--$1.6\times10^{-3}$\\
B1 & $3.3\times10^{-4}$--$1.4\times10^{-3}$\\
C0 & $9.3\times10^{-6}$--$1.6\times10^{-5}$\\
\bottomrule
\end{tabular}
\end{center}
These are training diagnostics, not certificates of the theorem's global hypotheses.

\subsection{Sampled reaction penalty}
For A1, define $q=F_U-\tfrac12\D F_S$. Its positive part is penalised as in
main-text \textup{(\mainref{eq:penalty})}, with $\D F_S$ evaluated by centred differences
on the DNS grid and one-sided differences at its ends. The reference set contains
$(U_{\rm ref},S_{\rm ref})$ and can be augmented by
$(\alpha U_{\rm ref},\alpha S_{\rm ref})$ for $\alpha\in\{0,1/4,1/2,3/4\}$.
The augmented objective retains $\mathcal L_{\rm fit}+30\mathcal L_{\omega,\rm ref}$
and adds $30\mathcal L_{\omega,\rm extra}$, averaged over the four added states and
all five Reynolds numbers. Only the penalty is applied there: no force targets are
assigned to zero or intermediate states. This computation requires no flow solve.
A finite penalty encourages, but need not eliminate, violations at sampled states;
it does not constrain every admissible profile.

\subsection{Integral construction and numerical quadrature}
For B1, the exact law of main-text \textup{(\mainref{eq:integralform})} has a non-negative
integrand $g=\nutan$ and therefore $\partial_S\tau=g(U,S)\ge0$.
The implementation uses eight-node Gauss--Legendre quadrature on $[0,1]$:
\begin{equation}
\tau_Q(U,S)=S\sum_{k=1}^8 w_k g(U,t_kS),\qquad
\partial_S\tau_Q=\sum_{k=1}^8w_k\big[g(U,t_kS)+St_k\partial_Sg(U,t_kS)\big].
\end{equation}
Positive weights preserve $\tau_Q S\ge0$, but do not alone fix the sign of the
quadrature derivative at every state. The assembled Jacobian differentiates this
implemented law. Independently of quadrature accuracy, the reaction contribution
$\tfrac12\partial_y\tau_U$ remains unconstrained. Thus the construction controls
only the ideal law's shear branch; the reported discrete indices assess the full
implemented tangent at the specified reference state.

\section{Nonlinear solver and preconditioner comparison}
\label{sec:precond}

\subsection{Newton--GMRES and stopping rules}
For the fixed residual of main-text \textup{(\mainref{eq:coupled})}, a Newton step solves
\begin{equation}
P^{-1}J(U^k)\delta U=-P^{-1}R(U^k),\qquad
U^{k+1}=U^k+\lambda\delta U.
\end{equation}
GMRES uses relative tolerance $10^{-10}$, restart length $\min(n,200)$ and at most
five restart cycles. The first step length in
$\{1,0.5,0.25,0.1,0.03,0.01,0.003\}$ that reduces the nonlinear residual is accepted;
if none does, the solve fails. The ensemble solves allow at most 60 Newton steps.
With integrated forcing $s_i=\dV_i$, define
\begin{equation}
\rho_R(U)=\frac{\|R(U)\|_{D^{-1}}}{\|s\|_{D^{-1}}}.
\label{eq:solverresidual}
\end{equation}
A solve succeeds if $\rho_R<10^{-11}$, or if two successive stagnation checks occur
below $10^{-8}$, each with the current residual exceeding $0.9$ times the best previous
value. This second criterion accommodates mesh-dependent round-off floors; it is
separate from the GMRES tolerance and from velocity accuracy.

The main ensemble uses $P=A$ with a zero initial profile and an ascending
Reynolds-number continuation. A successful profile is interpolated to the next mesh;
a failed rung resets the next start to zero. A case is credited if either start
converges, and the smaller velocity error is retained if both converge.

\subsection{Preconditioning changes work, not the root set}
The comparison uses the frozen-viscous preconditioner $P=A$, an augmented diffusion
preconditioner, and $P=J$. The augmented choice is $A+B_{\rm diff}$ for force closures,
$A+\mathcal D(\nu_{\rm tan})$ for stress closures and
$A+\mathcal D(\nut)$ for the hybrid. A nonsingular $P$ preserves the roots of $R$.
A Picard method instead approximates the Jacobian in the update; it also retains the
same target equation when converged, but may have different convergence behaviour.

To isolate the inner-solver comparison, these runs start from the DNS reference
profile, unlike the zero/continuation starts used for the main ensemble.
Table \ref{tab:precond} reports GMRES iterations per Newton step at $\Ret=1000$,
averaged over converged cases, and the relative spread in $\epsU$ across
preconditioners. Across the tested cases, the largest absolute difference in
$\epsU$ is $1.4\times10^{-12}$, while iteration counts differ by an order of magnitude.
This compares error metrics, not a direct norm of the difference between profiles.
It measures linear iteration work, not wall-clock cost or factorisation overhead.
Finite linear tolerances can alter convergence basins, so the observed agreement is
not a guarantee that different preconditioners always select the same root.

\begin{table}
\centering\small
\begin{tabular}{@{}lrrrr@{}}
\toprule
closure & $P{=}A$ & $P{=}A{+}B_{\rm diff}/{+}\nueff$ & $P{=}J$ & rel.\ spread in $\epsU$ \\
\midrule
A0  force & 45 $\pm$ 2 & 18 $\pm$ 2 & 1 $\pm$ 0 & 6.6e-12 \\
A1  force $+$ penalty & 28 $\pm$ 6 & 19 $\pm$ 1 & 1 $\pm$ 0 & 7.5e-12 \\
B0  stress & -- & -- & -- & -- \\
B1  stress, $\int\!\nu_{\rm tan}$ & 114 $\pm$ 6 & 20 $\pm$ 1 & 1 $\pm$ 0 & 8.9e-13 \\
C0  $\nu_t$ $+$ residual & 26 $\pm$ 5 & 28 $\pm$ 6 & 1 $\pm$ 0 & 2.3e-12 \\
\bottomrule
\end{tabular}


\caption{GMRES iterations per Newton step at $\Ret=1000$, mean $\pm$ s.d. over converged
reference-initialised runs. The final column is the maximum relative spread in
velocity error across preconditioners.}
\label{tab:precond}
\end{table}

\section{Indices, local estimates and realised amplification}
\label{sec:bound}

Table \ref{tab:indices} tabulates the two indices of main-text figure \mainref{fig:indices}.

\begin{table}
\centering\small\setlength{\tabcolsep}{3pt}
\begin{tabular}{@{}lrrrrr@{}}
\toprule
closure & 180 & 550 & 1000 & 2000 & 5200 \\
\midrule
\multicolumn{6}{@{}l}{$\omega(J)$, negative is dissipative} \\
A0 & 0.088 $\pm$ 0.0062 & 0.071 $\pm$ 0.013 & 0.23 $\pm$ 0.087 & 0.6 $\pm$ 0.3 & 1.9 $\pm$ 0.89 \\
A1 & \textbf{-0.033 $\pm$ 0.017} & \textbf{-0.041 $\pm$ 0.024} & \textbf{-0.037 $\pm$ 0.023} & \textbf{-0.019 $\pm$ 0.016} & \textbf{-0.0046 $\pm$ 0.0026} \\
B0 & 15 $\pm$ 15 & 2.7e+02 $\pm$ 3.5e+02 & 1.2e+03 $\pm$ 1.3e+03 & 4.3e+03 $\pm$ 1.6e+03 & 4e+04 $\pm$ 3.8e+04 \\
B1 & \textbf{-0.11 $\pm$ 0.011} & \textbf{-0.094 $\pm$ 0.0088} & \textbf{-0.091 $\pm$ 0.012} & \textbf{-0.088 $\pm$ 0.019} & -0.067 $\pm$ 0.058 \\
C0 & 1.1 $\pm$ 3.4 & 3.4 $\pm$ 10 & -0.0016 $\pm$ 0.0065 & -0.0017 $\pm$ 0.0054 & -0.0037 $\pm$ 0.0048 \\
\midrule
\multicolumn{6}{@{}l}{$\amp=\sigma_{\min}(A)/\sigma_{\min}(J)$, below one is better than frozen} \\
A0 & \textbf{0.76 $\pm$ 0.13} & 2.1 $\pm$ 1.6 & 3.4 $\pm$ 3.1 & 7.7 $\pm$ 11 & 10 $\pm$ 12 \\
A1 & \textbf{0.44 $\pm$ 0.12} & \textbf{0.14 $\pm$ 0.08} & \textbf{0.11 $\pm$ 0.13} & \textbf{0.12 $\pm$ 0.098} & \textbf{0.091 $\pm$ 0.042} \\
B0 & 25 $\pm$ 46 & 42 $\pm$ 46 & 1.4e+04 $\pm$ 4.2e+04 & 27 $\pm$ 52 & \textbf{0.12 $\pm$ 0.18} \\
B1 & \textbf{0.075 $\pm$ 0.01} & \textbf{0.028 $\pm$ 0.0015} & \textbf{0.015 $\pm$ 0.00044} & \textbf{0.0075 $\pm$ 0.00013} & \textbf{0.0027 $\pm$ 1.6e-05} \\
C0 & 5.5 $\pm$ 8.3 & 1.9 $\pm$ 2.4 & 1.6 $\pm$ 2.6 & \textbf{0.38 $\pm$ 0.2} & \textbf{0.25 $\pm$ 0.22} \\
\bottomrule
\end{tabular}


\caption{The two indices at $\uref$, mean $\pm$ s.d.\ over ten seeds. Upper block,
$\omega(J)$, bold where negative for every seed. Lower block, $\amp$, bold where below one.}
\label{tab:indices}
\end{table}

The table uses means and standard deviations; the main figure uses medians,
interquartile bands and extrema. The quoted improvement factors $372$ and $11$ are
$1/\langle\amp\rangle$ at $\Ret=5200$, not $\langle1/\amp\rangle$ and not relative
condition numbers.

\subsection{Local estimate and nonlinear remainder}

Let $x_0=D^{1/2}\uref$, $\hat e=x_*-x_0$, $\hat r=\mathcal R(x_0)$, and
$M=\mathcal R'(x_0)$ invertible. Suppose $\mathcal R'$ is Lipschitz with constant
$L$ in a convex neighbourhood containing the segment from $x_0$ to $x_*$.
Taylor's integral formula gives
\begin{equation}
0=\hat r+M\hat e+q,\qquad \|q\|_2\le\tfrac L2\|\hat e\|_2^2.
\end{equation}
Thus, with $\hat e_{\rm lin}=-M^{-1}\hat r$,
\begin{equation}
\|\hat e-\hat e_{\rm lin}\|_2
\le\frac{L}{2\sigma_{\min}(M)}\|\hat e\|_2^2,\qquad
\|\hat e_{\rm lin}\|_2\le\frac{\|\hat r\|_2}{\sigma_{\min}(M)}.
\label{eq:localbound}
\end{equation}
The second inequality bounds the linearised error exactly. It is not an unconditional
nonlinear bound, nor does an accuracy threshold control the remainder. By contrast,
the global monotonicity bound in part (i) of the main theorem has no remainder but requires a
uniform margin over all states. The experiments do not establish such a margin for A1
or B1.

For nonzero $\hat r$, define
\begin{equation}
\gamma:=\frac{\sigma_{\min}(M)\|M^{-1}\hat r\|_2}{\|\hat r\|_2}\in(0,1],
\qquad
\frac{\|\hat e_{\rm lin}\|_2}{\|\hat r\|_2}
=\frac{\gamma}{\sigma_{\min}(M)}.
\end{equation}
Equality in the upper bound occurs when the residual lies in the left singular
subspace corresponding to the smallest singular value. Small $\gamma$ means limited
excitation of the most sensitive directions; no assertion of exact orthogonality or
residual smoothness is required.

\subsection{Empirical comparison}

With $r=R(\uref)$, the measured amplification is
$\|U_*-\uref\|_D/\|r\|_{D^{-1}}$; the linear prediction is
$\|M^{-1}\hat r\|_2/\|\hat r\|_2$ and the upper estimate is $1/\sigma_{\min}(M)$.
Their comparison uses the deployed reference residual, not the regression loss.
Figure \ref{fig:bound} and table \ref{tab:ratio} compare the realised amplification
with $1/\sigma_{\min}(M)$ on the 105 solutions within the main text's accuracy threshold.
The realised amplification is below this linearised upper estimate in every such case;
the estimate/realised ratio is $105\pm201$, ranging from 7 to $1.2\times10^3$. Agreement with the linear
prediction suggests that directional alignment accounts for much of the looseness in
these cases, but this observation is not a proof that the nonlinear remainder is small
for other roots. The selection is by accuracy: A0 and B0 are absent because no accurate
root was found from the tested starts. The figure does not validate the global theorem
or exclude other roots.

\begin{figure}
\centerline{\includegraphics[width=\textwidth]{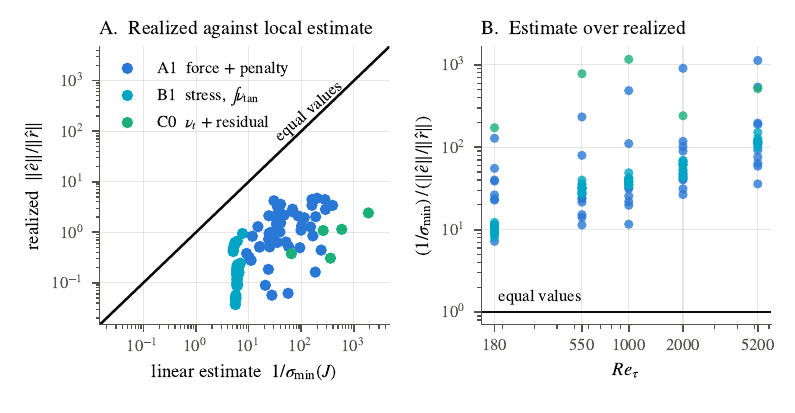}}
\caption{(\textit{a}) Realised amplification against the linearised upper estimate at the 105 accurate
solutions selected from the main ensemble, all below the diagonal. (\textit{b}) The ratio of the linearised upper estimate to the realised
amplification.}
\label{fig:bound}
\end{figure}

\begin{table}
\centering\footnotesize\setlength{\tabcolsep}{3pt}
\begin{tabular}{@{}lrrrrrr@{}}
\toprule
closure & $\Ret$ & $\|\hat r\|$ & estimate $1/\sigma_{\min}$ & linear & realized & $\gamma^{-1}$ \\
\midrule
A0 & 180 & 0.091 $\pm$ 0.011 & 56 $\pm$ 9.9 & 2.5 $\pm$ 1.3 & -- & -- \\
A0 & 550 & 0.089 $\pm$ 0.014 & 4.7e+02 $\pm$ 3.6e+02 & 17 $\pm$ 16 & -- & -- \\
A0 & 1000 & 0.12 $\pm$ 0.0098 & 1.4e+03 $\pm$ 1.3e+03 & 44 $\pm$ 52 & -- & -- \\
A0 & 2000 & 0.15 $\pm$ 0.012 & 6.3e+03 $\pm$ 8.6e+03 & 2.3e+02 $\pm$ 4.1e+02 & -- & -- \\
A0 & 5200 & 0.2 $\pm$ 0.017 & 2.2e+04 $\pm$ 2.5e+04 & 4.6e+02 $\pm$ 5e+02 & -- & -- \\
A1 & 180 & 0.12 $\pm$ 0.031 & 33 $\pm$ 8.7 & 2 $\pm$ 1.7 & 1.8 $\pm$ 1.3 & 36 $\pm$ 34 \\
A1 & 550 & 0.15 $\pm$ 0.059 & 31 $\pm$ 18 & 1.2 $\pm$ 0.75 & 1.3 $\pm$ 0.89 & 49 $\pm$ 64 \\
A1 & 1000 & 0.17 $\pm$ 0.053 & 46 $\pm$ 53 & 1.3 $\pm$ 1.3 & 1.3 $\pm$ 1.3 & 82 $\pm$ 138 \\
A1 & 2000 & 0.23 $\pm$ 0.079 & 98 $\pm$ 79 & 1.7 $\pm$ 1.3 & 1.5 $\pm$ 1.1 & 141 $\pm$ 256 \\
A1 & 5200 & 0.34 $\pm$ 0.084 & 1.9e+02 $\pm$ 87 & 2.5 $\pm$ 2.5 & 2 $\pm$ 1.5 & 239 $\pm$ 324 \\
B0 & 180 & 9.5 $\pm$ 13 & 1.9e+03 $\pm$ 3.4e+03 & 1.6 $\pm$ 3.8 & -- & -- \\
B0 & 550 & 12 $\pm$ 7 & 9.3e+03 $\pm$ 1e+04 & 0.82 $\pm$ 1.1 & -- & -- \\
B0 & 1000 & 16 $\pm$ 8.6 & 5.7e+06 $\pm$ 1.7e+07 & 3.5 $\pm$ 8.3 & -- & -- \\
B0 & 2000 & 22 $\pm$ 12 & 2.2e+04 $\pm$ 4.2e+04 & 0.12 $\pm$ 0.23 & -- & -- \\
B0 & 5200 & 34 $\pm$ 20 & 2.4e+02 $\pm$ 3.8e+02 & 0.014 $\pm$ 0.019 & -- & -- \\
B1 & 180 & 0.21 $\pm$ 0.04 & 5.6 $\pm$ 0.74 & 0.58 $\pm$ 0.14 & 0.58 $\pm$ 0.14 & 10 $\pm$ 1 \\
B1 & 550 & 0.098 $\pm$ 0.013 & 6.1 $\pm$ 0.32 & 0.19 $\pm$ 0.028 & 0.19 $\pm$ 0.028 & 32 $\pm$ 4 \\
B1 & 1000 & 0.1 $\pm$ 0.0097 & 6.1 $\pm$ 0.18 & 0.16 $\pm$ 0.016 & 0.16 $\pm$ 0.016 & 39 $\pm$ 4 \\
B1 & 2000 & 0.11 $\pm$ 0.0061 & 6.1 $\pm$ 0.1 & 0.11 $\pm$ 0.019 & 0.11 $\pm$ 0.019 & 58 $\pm$ 9 \\
B1 & 5200 & 0.13 $\pm$ 0.0096 & 5.6 $\pm$ 0.034 & 0.048 $\pm$ 0.0051 & 0.048 $\pm$ 0.0051 & 118 $\pm$ 13 \\
C0 & 180 & 0.28 $\pm$ 0.39 & 4e+02 $\pm$ 6.2e+02 & 0.6 $\pm$ 0.67 & 0.38 & 158 \\
C0 & 550 & 0.17 $\pm$ 0.19 & 4.2e+02 $\pm$ 5.2e+02 & 1.3 $\pm$ 2.2 & 2.4 & 244 \\
C0 & 1000 & 0.14 $\pm$ 0.11 & 6.4e+02 $\pm$ 1e+03 & 2.1 $\pm$ 2.6 & 0.31 & 1240 \\
C0 & 2000 & 0.16 $\pm$ 0.12 & 3.1e+02 $\pm$ 1.7e+02 & 1.2 $\pm$ 1.1 & 1.1 & 224 \\
C0 & 5200 & 0.18 $\pm$ 0.054 & 5.3e+02 $\pm$ 4.5e+02 & 0.74 $\pm$ 1.1 & 1.1 & 507 \\
\bottomrule
\end{tabular}


\caption{The linearised upper estimate against what is realised, mean $\pm$ s.d.\ over seeds, by closure and $\Ret$.
The residual, estimate and linear columns use all seeds; realised amplification and
$\gamma^{-1}$ use only accurate roots. The latter is estimate/linear, whereas figure
\ref{fig:bound}(\textit{b}) shows estimate/realised.
A dash means no solution within the accuracy threshold was found at that entry.}
\label{tab:ratio}
\end{table}

\section{Fixed-checkpoint multistart experiment}
\label{sec:supp:multistart}

\subsection{Frozen models and shared initial profiles}
At $\Ret=1000$ on the 200-cell mesh, we freeze A0 seed 0, A1 seed 0, B0 seed 8,
B1 seed 0 and C0 seed 9. Each had a converged solution in the earlier ensemble.
The same fixed weights are used for every initial condition of a given closure.
This tests multiplicity of one discrete equation per closure, rather than variability
between different trained networks.

All closures receive an identical set of 10,000 initial profiles, generated with
random-number seed 20260913. The first two are zero and laminar flow,
$U_{\rm lam}=-A^{-1}s$. The remaining 9,998 use main-text \textup{(\mainref{eq:random-ic})}:
\begin{equation}
U_0(y)=a\,y(2-y)+b\,v_p(y),\qquad
v_p(y)=\frac{\sum_{j=1}^{16}\xi_jj^{-p}\sin(j\upi y/2)}
{\max_i|\sum_{j=1}^{16}\xi_jj^{-p}\sin(j\upi y_i/2)|},
\end{equation}
with independent $\xi_j\sim\mathcal N(0,1)$. They cycle through four distributions:
\begin{center}
\begin{tabular}{@{}lccc@{}}
\toprule
Profile type & $p$ & $a$ & $b$\\
\midrule
Parabolic base and smooth perturbation & 2 & $\mathcal U(0,60)$ & $\mathcal U(0,10)$\\
Signed smooth profile & 2 & 0 & $\mathcal U(-100,100)$\\
Large smooth profile & 2 & 0 & $\pm10^\eta$, $\eta\sim\mathcal U(2,3)$\\
More oscillatory profile & 1 & $\mathcal U(-50,50)$ & $\mathcal U(0,50)$\\
\bottomrule
\end{tabular}
\end{center}
The sign in the third distribution is equiprobable. The sine and parabolic functions
vanish at both walls; no DNS profile is supplied as an initial condition.

\subsection{Screening, refinement and root separation}
The solver and early exits are those of \S\ref{sec:precond}, with $P=A$.
Screening is capped at 20 Newton steps or 15 seconds per initial condition.
A candidate with $\rho_R<10^{-8}$ is compared with the roots already found.
A new candidate receives up to 60 refinement steps or 30 seconds, using the same
Newton implementation. Each retained representative must satisfy
\begin{equation}
\rho_R(U)<10^{-8},\qquad
\frac{\|\delta U\|_D}{\max(1,\|U\|_D)}<10^{-7},\qquad
J(U)\delta U=-R(U).
\end{equation}
Profiles are grouped if
$\|U-V\|_D/\max(1,\|U\|_D,\|V\|_D)<10^{-4}$.
The calculation is divided into shards, whose root representatives are merged using
the same criterion. All merged representatives are independently verified with a
dense Jacobian solve. A hit counts a low-residual profile assigned to a verified root;
it need not have triggered the original solver's success flag. No screening start
reached the wall-clock cap.

\begin{table}
\centering\small
\begin{tabular}{@{}lrrr@{}}
\toprule
Closure & Initial profiles & Root hits & Distinct roots\\
\midrule
A0 & 10,000 & 910 & 7\\
A1 & 10,000 & 5,617 & 1\\
B0 & 10,000 & 230 & 93\\
B1 & 10,000 & 4,024 & 1\\
C0 & 10,000 & 2,993 & 4\\
\bottomrule
\end{tabular}
\caption{Fixed-model search underlying main-text figure \mainref{fig:multimodal}.
Hit counts depend on the specified initial-profile distribution and stopping rules.}
\label{tab:multistart}
\end{table}
Multiple verified roots provide numerical evidence of discrete nonuniqueness; no
claim is made that every branch persists under mesh refinement or in the continuum.
One observed root does not establish uniqueness, and failed starts do not establish
nonexistence. This experiment is separate from the 250-case ensemble in main-text
table \mainref{tab:useful}.

\subsection{Velocity and stress visualisation}
Each curve in main-text figure \mainref{fig:multimodal} is one distinct root, not one
initial condition. All fields are shown on $y/h\in[0,2]$ with independent vertical
ranges. The displayed stress diagnostic is reconstructed from the velocity using
$\tau_{\rm rec}=1-y-\nu GU$, with the same cell-gradient stencil for DNS and computed
profiles. This fixes the integration constant in $\nu U'+\tau=C-y$ to the DNS
normalisation $C=1$. It is not a separate comparison with the network's stress output;
for asymmetric computed profiles, the wall-stress partition need not have this value.

Blue denotes the reference-like branches and brown the other roots. In C0, the root
with $3.90\,\%$ velocity error is brown and the closest root, with $0.476\,\%$ error,
is blue. This branch colouring differs from the uniform $5\,\%$ error threshold
used for the ensemble table and is not an additional physical-admissibility test.

\end{document}